\documentclass{article}

\PassOptionsToPackage{numbers, compress}{natbib}

\usepackage{arxiv}
\usepackage[utf8]{inputenc} 
\usepackage[T1]{fontenc}    
\usepackage{hyperref}       
\usepackage{url}            
\usepackage{booktabs}       
\usepackage{amsfonts}       
\usepackage{nicefrac}       
\usepackage{microtype}      
\usepackage{amsmath}
\usepackage{mathtools}
\usepackage{amssymb}
\usepackage{bm}
\usepackage{algorithm}
\usepackage{algorithmic}
\usepackage{siunitx}
\usepackage{multirow}
\usepackage{wrapfig}
\usepackage{subfig}
\usepackage[export]{adjustbox}
\usepackage{caption}
\usepackage{graphicx}
\usepackage[dvipsnames]{xcolor}
\usepackage{scalerel}
\usepackage{tikz}
\usepackage[inline]{enumitem}
\usepackage[english]{babel}     
\usepackage{amsthm}

\makeatletter

\DeclareMathOperator*{\argmin}{argmin}
\DeclareMathOperator*{\argmax}{argmax}
\DeclareMathOperator*{\EX}{\mathbb{E}}
\DeclareMathOperator*{\Bigcdot}{\scalerel*{\cdot}{\bigodot}}

\newcommand{\norm}[1]{\left\lVert#1\right\rVert}

\newcommand{\CI}[2]{#1\textcolor{gray}{\scaleto{\pm#2}{5pt}}}

\newcommand{\dblacc}[2]{\mathpalette\dblacc@{{#1}{#2}}}
\newcommand{\dblacc@}[2]{\dblacc@@#1#2}
\newcommand{\dblacc@@}[3]{%
  \begingroup
  \sbox\z@{$\m@th#1#3$}%
  #2{\box\z@}%
  \endgroup
}
\newcommand\inlineeqno{\stepcounter{equation}\hfill (\theequation)}
\newcommand{\uvec}[1]{\boldsymbol{\hat{\textbf{#1}}}}
\newcommand{\pseudoeqref}[1]{\textup{\tagform@{#1}}}
\newtheorem*{theorem*}{Theorem}

\makeatother

\title{C-SL: Contrastive Sound Localization with Inertial-Acoustic Sensors}

\author{%
    Majid Mirbagheri \\
    University of Washington \\
    \texttt{mbagheri@uw.edu} \\
 \And
 Bardia Doosti \\
 Indiana University Bloomington \\
 \texttt{bdoosti@indiana.edu} \\
}

\begin{document}

    \maketitle

    \begin{abstract}
        Human brain employs perceptual information about the head and eye movements to update the spatial relationship between the
        individual and the surrounding environment.
        Based on this cognitive process known as spatial updating, we introduce contrastive sound localization (C-SL) with mobile
        inertial-acoustic sensor arrays of arbitrary geometry.
        C-SL uses unlabeled multi-channel audio recordings and inertial measurement unit (IMU) readings collected
        during free rotational movements of the array to learn mappings from acoustical measurements to an array-centered
        direction-of-arrival (DOA) in a self-supervised manner.
        Contrary to conventional DOA estimation methods that require the knowledge of either the array geometry or source locations in the
        calibration stage, C-SL is agnostic to both, and can be trained on data collected in minimally constrained settings.
        To achieve this capability, our proposed method utilizes a customized contrastive loss measuring the spatial contrast between
        source locations predicted for disjoint segments of the input to jointly update estimated DOAs and the acoustic-spatial mapping
        in linear time.
        We provide quantitative and qualitative evaluations of C-SL comparing its performance with baseline DOA estimation methods in a
        wide range of conditions.
        We believe the relaxed calibration process offered by C-SL paves the way toward truly personalized augmented hearing applications.

    \end{abstract}

    \section{Introduction}\label{sec:intro}

    Humans localize sounds by comparing inputs across the two ears, resulting in a head-centered representation of
    sound-source location~\cite{blauert1997spatial}.
    When the head moves, brain combines inertial information about head movement with the head-centered estimate to correctly update
    the world-centered sound-source location in a cognitive process known as auditory spatial
    updating~\cite{wallach1940role,genzel2016dependence}.
    Existing methods for sound localization with microphone arrays differ from human auditory system in two major aspects:
    \begin{enumerate*}[label=(\roman*)]
        \item unlike humans who adapt to changes in auditory localization cues without supervision~\cite{shinn1998adapting}, in
        order to operate, these algorithms rely on either specific array geometries or access to sample sounds with known spatial information.
        \item these methods do not account for array movements as they are mostly designed for static applications.
    \end{enumerate*}
    With the advent of augmented reality (AR) technologies embodied in mobile devices such as smart glasses and headphones, addressing these
    gaps can extend versatility of these algorithms to more applications in this domain.

    Calibration of the arrays in conventional source localization methods involves measuring array responses to signals coming from known
    directions when they cannot be analytically determined as a function of the array geometry.
    Once the array is calibrated, these methods use stored responses as some form of lookup table.
    Popular methods in this category consist of those based on steered response power (SRP)~\cite{dibiase2000high} and subspace approaches
    such as multiple signal classification (MUSIC)~\cite{schmidt1986multiple}.
    The grid search involved in these methods is, however, usually of considerable computational cost, while the performance is restricted
    by the grid resolution.

    In an effort to overcome these issues, more recently supervised learning algorithms using deep neural networks (DNN) have gained
    significant attention in the field~\cite{roden2015sound, xiao2015learning,chakrabarty2017broadband,wang2018robust,
    adavanne2018direction,comminiello2019quaternion}.
    Given acoustic measurements with known spatial labels, in the form of a single direction or spherical intensity field representations,
    DNN-based methods solve a nonlinear regression problem to predict labels from measurements via an iterative gradient-based
    optimization algorithm.
    A common problem of learning-based methods is their sensitivity to mismatches between distributions of data used for training and test.
    This issues can be specifically more severe for mobile arrays with microphones that are fit in the ear or installed on head-mounted
    or hand-held devices.
    The directional pattern of such arrays depends on not only relative positioning of the microphones but also the unique anatomical
    geometry of the users' head/ear, the device fit, or how the device is handled by the user.
    On one hand, augmenting training sets with all such variabilities is in general an infeasible task, and this eliminates the
    possibility of calibrating the array prior to deployment.
    On the other hand, collection of acoustic data with clean spatial labels cannot currently take place on a per user basis as it requires
    elaborate lab setups or computationally-expensive simulations.

    Contrastive learning is an emerging paradigm proposed to overcome data limitations of supervised methods through self-supervision
    namely automatic labeling of data by comparing different views of it across time, space, or sensor
    modalities~\cite{oord2018representation,chen2020simple,tian2019contrastive}.
    This paradigm has been successfully used for visual object detection~\cite{sermanet2018time,pirk2019online}, and audio-visual source
    localization~\cite{zhao2018sound,owens2018audio,gan2019self,liu2019self}.
    Studies in neuroscience suggest that human brain utilizes predictive coding, a special form of self-supervision, to encode sound
    attributes~\cite{kumar2011predictive}.
    Spatial updating process in brain also by nature uses a contrastive measure based on spatial displacement of the head to update
    head-centered sound source location as the individual moves~\cite{wallach1940role,genzel2016dependence}.
    Inertial information involved in calculation of head attitude and motion is provided by vestibular organs in the inner ear.
    While this process has been mostly investigated in the context of localization inference, a contrastive learning framework for sound
    localization based on spatial updating that imitates plasticity in spatial auditory processing is yet to be desired.
    Such a framework will bridge the gaps between traditional DOA estimation methods and human spatial auditory processing.
    In applications, inertial information has been previously utilized to increase robustness of visual
    odometry~\cite{almalioglu2019selfvio,qin2017robust} and simultaneous localization and mapping (SLAM) systems~\cite{mur2017visual}.

    \textbf{Contributions} In this paper, we propose to the best of our knowledge the first contrastive learning framework for sound
    localization with inertial-acoustic sensors based on cognitive process of auditory spatial updating.
    Our algorithm, named \textit{C-SL}, is able to localize both narrowband and wideband sources, and in contrast to existing DNN and grid
    search-based methods, is agnostic to the array geometry and the knowledge of source locations in the calibration process.
    The only assumption we make is that during training there is only one far-field source present, and that the location of this source is
    approximately piece-wise constant in a reference coordinate frame which we refer to as \textit{world-frame} in the rest of the paper.
    To train our model, we use a customized loss that leverages this assumption and minimizes spatial contrast between predictions for
    consecutive segments of input in the calibration stage.
    In the next section, we describe the data model followed by how the contrastive loss is computed, and the model architecture.

    \section{Self-Supervised Learning with Sub-Contrastive Loss}\label{sec:ss-learning}

    \subsection{Data Model}\label{subsec:data-model}
    Assuming a single far-field sound wave impinging on a microphone array, the output of the microphone with the index
    $m\in\{1, \dots, M\}$ is given in the short-time Fourier transform (STFT) domain by:

    \begin{equation}
        \label{eq:eq1}
        Y^m_{k,n}=H^m_k(r^s_n)S_{k,n}+V^m_{k,n}
    \end{equation}

    where $S_{k,n}$ is the source signal, $H^m_k(r^s_n)$ is the acoustic transfer functions (ATF) of the source at location $r^s_n$ with
    respect to $m$-th microphone, $V^m_{k,n}$ models noise and reverberation, and $k$ and $n$ are the frequency and time frame indices,
    respectively.
    Since we are interested about far-field localization, we denote the source location as a 3-D vector on the unit sphere,
    $r^s\in \mathbb{S}^2$.
    With this definition, the locations will be the same for all microphones, hence referred to as \textit{sensor-frame} direction.
    Throughout the paper, bold symbols represent $M$-dimensional vectorized version of quantities related to the microphone array,
    $\langle \cdot{,}\cdot \rangle$ is the inner product, $\norm{x}$ denotes $\ell_2$ norm of a vector,
    and $\hat{x}=\frac{x}{\norm{x}}$ for all vectors $x\ne0$.

    During train data collection the array is rotated in all directions to densely sample acoustic measurements along arbitrary
    trajectories on $\mathbb{S}^2$.
    With a 9-DOF inertial measurement unit (IMU) attached to the array, orientations of the array with respect to the earth (world) frame,
    represented by quaternions or Euler angles, can be calculated from raw IMU readings~\cite{madgwick2010efficient}.
    Given the correspondence between orientations and rotation matrices in 3-D space~\cite{kuipers1999quaternions}, we assume that for
    any given time frame we know the corresponding rotation matrix $R_n\in\mbox{SO}(3)$ with which we can transform any direction in
    the sensor coordinate to the world frame coordinate by:

    \begin{equation}
        \label{eq:trans}
        r^w_{n} = R_{n} r^s_{n}
    \end{equation}

    \textbf{Spatial Constancy:} Considering that $r^w$ changes at a slow rate (in contrast to $r^s$), we assume it to be (approximately)
    constant over time intervals, denoted by $\{I_i\}_{i=1}^{N_i}$, with
    $I_{i}=\{n\}_{n_i\leq n<n_{i+1}}$, $1=n_1<n_2<\dots<n_{N_i+1}=N_s$, and $N_s$, $N_i$ representing the total number of samples and
    intervals.
    For the sake of generality, we do not assume any special relationship between source locations across different intervals.

    \subsection{Sub-Contrastive Loss}\label{subsec:gc-loss}

    Given the observations $\{(\bm{Y}_{k,n},R_n)\}$, and $\{I_i\}$, we seek an acoustic-spatial function
    $f_\theta \colon \mathbb{C}^M\times\left(0,1\right]\rightarrow \mathbb{S}^2$, parameterized by $\theta$, that maps $M$-dimensional
    complex-valued acoustic measurements at each time-frequency bin, $\bm{Y}_{k,n}$, and their associated normalized frequency,
    $\tilde{k}=k/k_{\max}$, to a single direction in the sensor frame.

    We can find the optimum values of $\theta$ in a self-supervised manner by leveraging the spatial constancy assumption i.e.\ maximizing
    pairwise similarity between world-frame directions predicted for all the time-frequency bins within intervals $\{I_i\}$ with a
    contrastive loss expressed by:

    \begin{align}
        \label{eq:l_cont}
        \mathcal{L}_\mathrm{cont}(\theta)=\sum\limits_{i=1}^{N_i} \
        \sum\limits_{\substack{m,n \in I_{i}\\k, k'}} \
        \norm{\tilde{r}^w_{k, m}-\tilde{r}^w_{k', n}}^2
    \end{align}

    where $\tilde{r}^w_{k, n}=R_{n}f_\theta (\bm{Y}_{k, n}, \tilde{k}). \inlineeqno$\label{eq:w_pred}
%

    This loss function uniquely determines the mapping $f$ up to an inversion of the sign, when the mapping is
    bijective and the loss takes its minimum value of 0 over all possible pairs of observations (for a proof see the Appendix).
    In a conventional contrastive learning framework, the loss function not only encourages outputs to be close for similar (positive)
    examples, but also forces them away from one another for distant (negative) ones.
    The loss defined in~\eqref{eq:l_cont} accomplishes this in a soft manner by constraining the spatial contrast between
    sensor-frame directions predicted for different pairs based on the measured change in the orientation of the of the array.

    Pairwise similarity imposes a fairly strong constraint on predicted directions causing the training to fail entirely when some of
    the bins are dominated by spurious directions caused by reverberation, or ambient noise.
    In order to manage such situations, we propose a simple cost function, termed \textit{sub-contrastive} loss, that
    enforces a weaker similarity constraint over the entire set of interval bins.
    Furthermore, to handle uncertainty, we expand the range of $f_\theta$ to $\mathbb{R}^3$ so that the mapping predicts an additional
    positive-valued weight, encoded in the norm of the output, representing the confidence of the predicted direction for the bin.

    Figure~\ref{fig:overview} shows how sub-contrastive loss is computed from world-frame directions predicted for different
    time-frequency bins.
    We first break each interval $I_i$ into two disjoint sub-intervals denoted by $\{(I_{i,1}, I_{i,2})\}$.
    Break points are chosen randomly so that the ratio of bins in the two segments satisfies
    $c_1\leq\frac{\lvert I_{i,1}\rvert}{\lvert I_{i,2}\rvert}\leq c_2$ with $c_1$ and $c_2$ arbitrarily set to 0.2 and 0.8.

    We then pool the world-frame directions predicted by~\pseudoeqref{4} over all time-frequency bins in each sub-interval by finding
    their centroid on $\mathbb{S}^2$ formulated as:

    \begin{equation}
        \label{eq:centroid}
        \bar{r}_{i,l}^w=\argmin_{d\in \mathbb{S}^2}\sum\limits_{n\in I_{i,l},k}\norm{d-\tilde{r}^w_{k, n}}^2
        =\dblacc\widehat{\left(\sum\limits_{n\in I_{i,l},k}\tilde{r}^w_{k, n}\right)},\quad l=1,2
    \end{equation}

    Note that by preserving the norms during sensor-to-world transformation, we favor the high-confidence predictions over the rest in
    the pooling stage.
    Finally, the sub-contrastive loss measures the distance between the two centroids computed for the sub-interval pairs as followed:

    \begin{equation}
        \label{eq:l_sub_cont}
        \mathcal{L}_{\text{sub-cont}}(\theta)=\sum\limits_{i=1}^{N_i} \
        \norm{\bar{r}_{i,1}^w-\bar{r}_{i,2}^w}^2
    \end{equation}

    A minimum value of $0$ for the sub-contrastive loss is only a necessary condition to meet the pairwise similarity constraint.
    Thus the this loss can be seen as a weaker form of the contrastive one.
    It should be noted that using random subsets of bins from intervals to compute the centroids results in same degenerate sensor predictions.
    The time-based segmentation described above avoids this situation by considering only subsets with likely most different array
    orientations and making sure their world frame centroids are matching.
    Randomizing break points across training epochs improves stochasticity and consequently generalizability of the model.

    Similar to contrastive one, the sub-contrastive loss is invariant to reflection of predicted sensor frame directions with respect to
    the origin.
    This sign ambiguity in the predictions can be easily resolved in many cases via a postprocessing stage in which $f_\theta$ is negated
    based on additional criteria.
    For instance, when simple knowledge about relative position of microphones such as
    ``mic A coordinate in the sensor frame has a larger value on the $x$-axis than mic B'' is available, a comparison of intensities
    or delays of sounds received at the two microphones can determine if the predicted directions should be reflected or not.
    Alternatively, orientation of the array at the beginning of data collection can be set with respect to the source in a way that a
    general condition such as $\langle r^s_{n=0},\uvec{i}\rangle > 0$ is enforced and later used to disambiguate the mapping.

    The centroids computed in~\eqref{eq:centroid} can be interpreted as \textit{denoised} approximation of predicted world-frame directions.
    At the time of training, we need two versions of this estimate for each interval to make contrastive learning possible.
    However at the time of inference, there is no such need, and centroids can be computed directly over all time-frequency bins within the
    each interval providing that they all belong to the same source.
    In multi-source conditions, world-frame predictions for different bins appear in clusters representing different sources
    to which they belong.
    As we will see in Section~\ref{sec:exp}, in such situations, we can run a clustering algorithm on these predictions in the same vein
    as~\cite{wu2018multisource}, and use associated cluster centers as denoised approximation of world-frame directions for each bin.
    Regardless of number of sources, denoised sensor-frame directions are computed by transforming the world-frame centroids back into
    the sensor frame for each time-frequency bin.

    In contrast to quadratic time of contrastive loss, computation of the sub-contrastive loss only takes linear time with respect
    to the number of time-frequency bins resulting in a very efficient training of the model by C-SL\@.

    \subsection{Model Architecture}\label{subsec:model-arch}

    We design $f_{\theta}$ as a multi-layer perceptron (MLP), agnostic to the underlying spatial configuration of the array (as opposed to a
    convnet, for example).
    The MLP we use consists of three hidden layers with 1024, 512, and 256 units.
    Each hidden layer is followed by a (parameter-free) weight normalization layer~\cite{salimans2016weight} and a standard ReLU non-linearity.
    The third hidden layer provides the input to a linear prediction layer of size three.

    It should be noted that while $f_\theta$ could be optimized separately for each frequency, we opt for a single mapping conditioned on
    frequency in the view of the fact that the array spectral profile is inherently low dimensional.

    \section{Experiments}\label{sec:exp}

    \begin{table}[t]
        \centering
        \begin{tabular}{cc}
            \hspace*{-5mm}
            \begin{minipage}{0.5\textwidth}
                \begin{algorithm}[H]
                    \caption{C-SL Training}
                    \label{alg:C-SL}
                    \begin{algorithmic}[1]
                        \STATE $\theta$ $\gets$ Initialize model parameters.
                        \WHILE{not converged}
                        \STATE $B\subset\{1, \dots, N_i\}$ $\gets$ random mini-batch of interval indices
                        \STATE $\bm{Y}_{k,n}$, $R_n$, $\gets$ data at intervals $\{I_i\}_{i\in B}$
                        \STATE $\{(I_{i,1},I_{i,2})\}\gets$ sub-intervals with random break points
                        \STATE $\tilde{r}^s_{k, n}\gets f_\theta (\bm{Y}_{k, n}, \tilde{k})$
                        \STATE $\tilde{r}^w_{k, n}\gets R_{n}\tilde{r}^s_{k, n}$
                        \STATE $\bar{r}_{i,l}^w\gets\dblacc\widehat{\left(\sum\limits_{n\in I_{i,l},k}\tilde{r}^w_{k, n}\right)}\quad i\in B,\,l=1,2$
                        \STATE $\mathcal{L}_{\text{sub-cont}}\gets\sum\limits_{i\in S} \
                            \norm{\bar{r}_{i,1}^w-\bar{r}_{i,2}^w}^2$
                        \STATE $\theta\gets\text{A}\textsc{dam}(\triangledown_\theta,\mathcal{L}_{\text{sub-cont}},\theta)$
                        \ENDWHILE
                        \IF{reflection condition satisfied}
                        \STATE $f_\theta\gets -f_\theta$
                        \ENDIF
                    \end{algorithmic}
                \end{algorithm}
            \end{minipage} &
            \begin{minipage}{0.5\textwidth}
                \centering
                \captionsetup{width=.9\linewidth}
                \includegraphics[width=\textwidth]{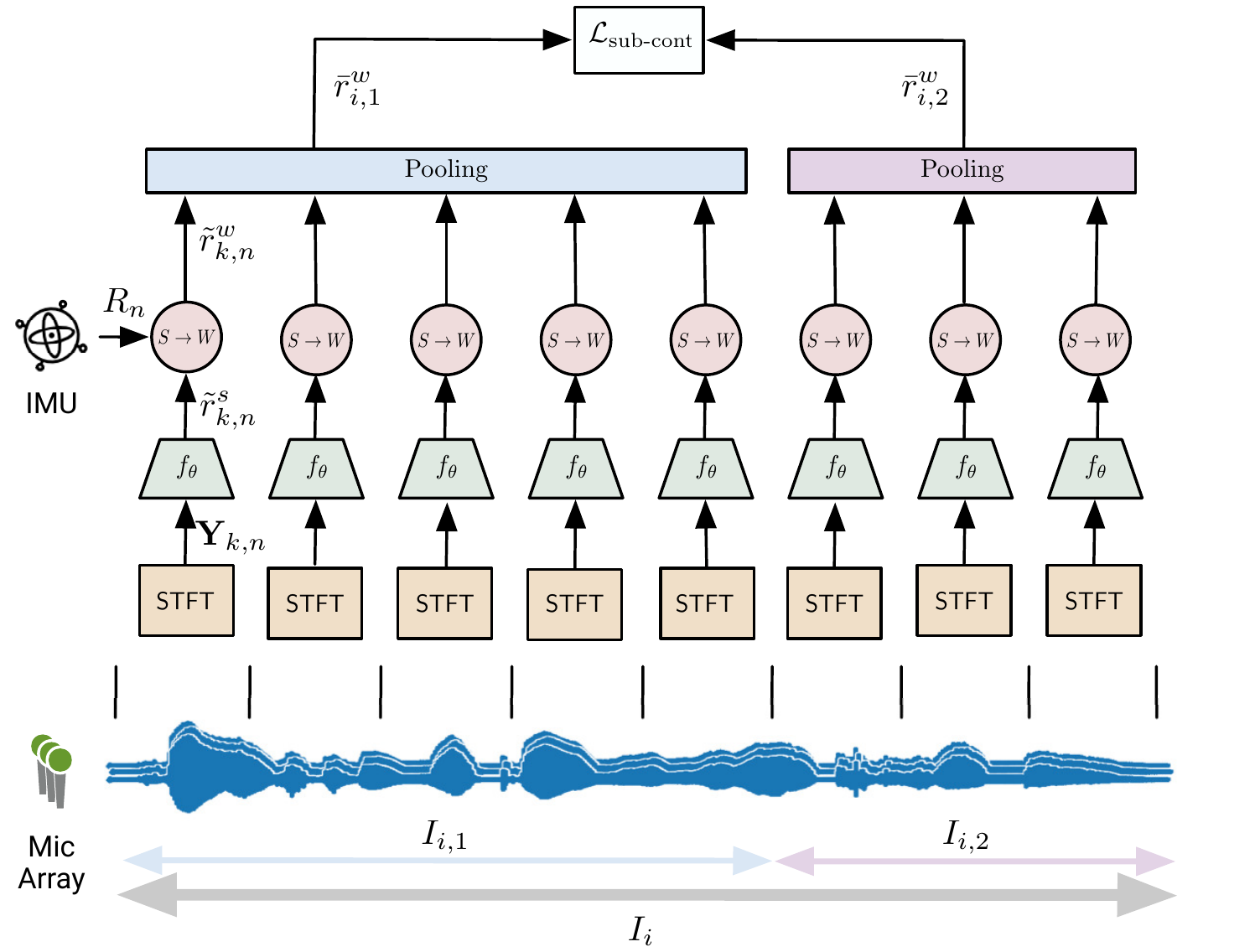}
                \captionof{figure}{Overview of the proposed framework for Contrastive Sound Localization (C-SL).
                Sensor-frame predictions are transformed into the world frame ($S\rightarrow W$) by the help of inertial information
                provided by the IMU.
                During training, sub-contrastive loss measures distance between world-frame predictions aggregated over
                time-frequency bins within sub-intervals of data. }
                \label{fig:overview}
            \end{minipage}
        \end{tabular}
    \end{table}

    \subsection{Dataset}\label{subsec:dataset}
    In order to evaluate the proposed C-SL, we synthesize a dataset consisting of hybrid acoustic-inertial data in the same vein as most
    DNN-based methods which need large amount of data~\cite{roden2015sound, xiao2015learning,chakrabarty2017broadband,adavanne2018direction}.
    Our dataset consists of multiple recording sessions.
    Each session simulates acoustic-inertial data from one interval described in section~\ref{subsec:data-model}.
    We simulate measurements in a \SI[parse-numbers=false]{4\times4\times4}{m} room with point sources randomly placed \SI{1}{m}
    away from the room center.
    Source locations remain consistent within a particular session, but vary from session to session.
    Without loss of generality, we use an array with cubical configuration with 8 omni-directional microphones positioned at
    the corners of a cube with an edge length of \SI{5}{cm} and center of mass always at the room center.
    For the array motion, we consider rotations at a constant angular velocity of magnitude
    \SI[parse-numbers=false]{\frac{\pi}{2}}{\radian/\second} and fixed but random axes for each session.
    Orientation of the array at the beginning of each session is set to a unit quaternion randomly chosen from $\mathbb{S}^3$.
    Translational motion was not considered in our experiments in order to avoid violating the far-filed assumption.
    We calculated room impulse responses in five different reverbrant room conditions, four with fixed reverberation times, and one with
    mixed reverberation times.
    In all conditions, we used a GPU-based geometrical acoustics simulator, gpuRIR~\cite{diaz2018gpurir} to model sound propagation and
    reverberation based on the rectangular room image-source model~\cite{allen1979image}.
    In the first four conditions, we set the value of broadband reverberation time of the room $T_{60}$ respectively to 0 (anechoic), 250,
    500, and \SI{750}{ms}.
    For the last condition, we randomly sampled $T_{60}$ values from the range $\SIrange[range-phrase=-]{0}{750}{ms}$ for
    each session.
    We refer to these five conditions as $C_{\text{anechoic}}$, $C_{250}$, $C_{500}$, $C_{750}$, and $C_{\text{mixed}}$.
    To generate one session with a single source in each of these conditions, one dry speech recording from the TIMIT corpus~\cite{garofolo1993darpa}
    with an average length of approximately \SI{3}{s} was convolved with the simulated room impulse responses.
    We split 6300 utterances in the corpus into three subsets of size 5012, 638, and 638 respectively for train, test and validation.
    All utterances in the three splits were then spatialized each resulting in one session in the corresponding split.
    All recordings, sampled at \SI{16}{k Hz}, were transformed into the STFT domain using frames of length \SI{25}{ms}, hop length
    \SI{10}{ms} and a hanning window.
    Rotation matrices associated with time frames were calculated based on angular velocity and initial orientation of the
    array chosen for the session and the timestamp of the frames.

    The input features to the network are $17\times1$ vectors computed for each time-frequency bin, consisting of real and imaginary parts
    of the array STFT coefficients concatenated and normalized to a unit-norm $16\times1$ vector and an additional feature
    representing normalized frequency of the bin.
    The normalization of acoustic features instructs the model to disregard content and distance-related variations of intensity across
    time and frequency.

    In order to save memory and also satisfy the far-filed assumption, during training from each session we only picked bins in the
    frequency range \SIrange{340}{8000}{Hz} whose original STFT magnitudes were greater than some ratio (arbitrarily set to -40 dB)
    of the maximum magnitude over all the bins with the same frequency.

    \subsection{Training and Metric}\label{subsec:train_and_metric}

    We optimized parameters of the model iteratively on selected mini-batches of simulated sessions from $C_\text{mixed}$ dataset.
    We used the Adam optimizer~\cite{kingma2014adam} with $\beta_1=0.9$ and $\beta_2=0.999$, a learning rate of
    \num[output-exponent-marker=\ensuremath{\mathrm{e}}]{1e-05}, and batch size of 8.
    All models were trained with 2 GPUs for 300 epochs.
    The procedural training details are summarized in Algorithm~\ref{alg:C-SL}.

    We used the angle (in degrees) between estimated sensor-frame directions and their ground truth values used in the simulation,
    formulated by $\sigma(\bar{r}^s,r^s)=180/\pi\cdot\cos^{-1}(\langle\bar{r}^s,r^s\rangle)$, as the DOA estimation
    error metric in our evaluations.
    In particular for C-SL, $\bar{r}^s$ refers to the final denoised sensor-frame estimates.

    \begin{table}[t]
        \caption{Comparison of DOA estimation errors (in degree) in single-source condition evaluated for different reverberations
        times, and window lengths $L_{\text{win}}$.}
        \label{table:comparison}
        \centering
        \begin{tabular}{@{}llcccc@{}}
            \toprule
            Method & $L_{\text{win}}(\si{s})$ & $C_{\text{Anechoic}}$ & $C_{250}$ & $C_{500}$ & $C_{750}$ \\
            \midrule
            \multirow{5}{2.5cm}{SRP-PHAT~\cite{dibiase2000high}} \
            & 0.05 & $\CI{1.17}{0.01}$ & $\CI{3.28}{0.12}$ & $\CI{11.35}{0.30}$ & $\CI{16.96}{0.39}$ \\
            & 0.2 & $\CI{1.16}{0.01}$ & $\CI{1.96}{0.06}$ & $\CI{4.16}{0.19}$ & $\CI{6.67}{0.35}$ \\
            & 0.5 & $\CI{1.27}{0.03}$ & $\CI{2.41}{0.11}$ & $\CI{4.59}{0.27}$ & $\CI{6.63}{0.43}$ \\
            & 1.0 & $\CI{2.30}{0.09}$ & $\CI{3.80}{0.16}$ & $\CI{7.15}{0.44}$ & $\CI{9.11}{0.60}$ \\
            & Full (\textasciitilde3) & $\CI{9.76}{0.64}$ & $\CI{10.78}{0.67}$ & $\CI{14.59}{0.99}$ & $\CI{16.04}{1.08}$ \\
            \midrule
            \multirow{5}{2.5cm}{LSDD~\cite{tourbabin2019lsdd}} \
            & 0.05 & $\CI{1.12}{0.01}$ & $\CI{5.22}{0.19}$ & $\CI{17.81}{0.39}$ & $\CI{24.98}{0.48}$ \\
            & 0.2 & $\CI{1.20}{0.01}$ & $\CI{2.12}{0.06}$ & $\CI{5.70}{0.28}$ & $\CI{10.14}{0.48}$ \\
            & 0.5 & $\CI{1.30}{0.03}$ & $\CI{2.32}{0.13}$ & $\CI{5.39}{0.37}$ & $\CI{8.08}{0.54}$ \\
            & 1.0 & $\CI{1.75}{0.04}$ & $\CI{2.84}{0.24}$ & $\CI{6.01}{0.50}$ & $\CI{8.95}{0.65}$ \\
            & Full (\textasciitilde3) & $\CI{6.76}{0.46}$ & $\CI{8.32}{0.64}$ & $\CI{12.79}{0.95}$ & $\CI{14.50}{0.98}$ \\
            \midrule
            \multirow{5}{2.5cm}{C-SL (proposed)} \
            & 0.05 & $\CI{1.56}{0.01}$ & $\CI{7.91}{0.16}$ & $\CI{18.63}{0.30}$ & $\CI{25.03}{0.38}$ \\
            & 0.2 & $\CI{1.42}{0.02}$ & $\CI{3.47}{0.07}$ & $\CI{7.61}{0.22}$ & $\CI{11.25}{0.37}$ \\
            & 0.5 & $\CI{1.25}{0.02}$ & $\CI{2.94}{0.08}$ & $\CI{5.97}{0.31}$ & $\CI{7.93}{0.42}$ \\
            & 1.0 & $\CI{1.17}{0.02}$ & $\CI{2.78}{0.12}$ & $\CI{5.47}{0.41}$ & $\CI{7.18}{0.55}$ \\
            & Full (\textasciitilde3) & $\CI{1.03}{0.03}$ & $\CI{2.29}{0.10}$ & $\CI{3.86}{0.21}$ & $\CI{4.66}{0.25}$ \\
            \bottomrule
        \end{tabular}
    \end{table}

    \subsection{Evaluation Results}\label{subsec:baseline}

    The distinguishing characteristics of C-SL is its self-supervised nature and the new applications made possible because of that, most
    notably when source locations are not available for array calibration.
    In order to demonstrate this capability, we run C-SL under a wide range of conditions and compared its performance with
    two baseline methods that leverage knowledge of array transfer functions: the well-studied SRP-PHAT~\cite{dibiase2000high}
    algorithm, and another approach, named LSDD with soft time-frequency masks~\cite{tourbabin2019lsdd}, recently proposed for
    highly-reverbrant environments.\footnote{
    Existing DNN-based methods could not be trained on our dataset since they required both source and array to be stationary for
    at least several seconds.}
    In summary, SRP-PHAT estimates the sensor directions by the maximum of the normalized cross-power spectral density (CPSD),
    steered in all possible directions $\{r_j\}_{j=1}^J$ i.e.\
    $\bar{r}^s_{n, \text{SRP}}=\argmax\limits_{r_j}\sum\limits_k\left|\sum\limits_{m=1}^M {A^m_k}^\ast(r_j)\frac{\Phi_{k,i}^{m,1}}{\rvert
        \Phi_{k,i}^{m,1}\rvert}\right|^2$
    where $A^m_k(r_j)=H^m_k(r_j)/H^1_k(r_j)$, and $\Phi^{m,1}_{k, n}=\EX(Y^m_k{Y^1_k}^\ast)$ is the cross-power spectral density
    between the $m$-th and first microphone signals estimated for a window $I_n$ centered at time index $n$.
    LSDD method directly uses similarity between array outputs and precomputed ATFs weighted by a mask measuring direct path dominance
    to estimate DOAs.
    In particular, it first computes a spatial spectrum for each bin, $\phi_{k,n}(r_j)=\arccos{\frac{\lvert\langle
        \bm{H}_k^H(r_j),\bm{Y}_{k,n}\rangle\rvert}{\norm{\bm{H}_k(r_j)}\norm{\bm{Y}_{k,n}}}}$
    where $(\cdot)^{H}$ is the Hermitian transpose.
    Soft masks are then calculated by $w_{k,n}=\min\limits_j{\phi_{k,n}(r_j)}$.
    Finally, it finds sensor-frame direction for each interval through a grid search: $\bar{r}^s_{n, \text{LSDD}}=\argmin\limits_{r_j}\sum
    \limits_{k,n'\in I_n}w_{k,n'}^{-\alpha}\phi_{k,n'}(r_j)$ in which $\alpha>0$ is a selectivity factor.
    We used a uniform grid of \ang{2} resolution consisting of 13744 directions for both LSDD and SRP-PHAT\@.
    Both of these methods estimate DOAs for wideband sources i.e.\ they use moving windows to estimate one direction for the center time frame.
    We found out SRP-PHAT performed best when time frames were extracted in the frequency range \SIrange{340}{6000}{Hz}.
    For LSDD, we chose the frequency range \SIrange{1800}{3600}{Hz} and selectivity factor $\alpha=3$ as prescribed
    in~\cite{tourbabin2019lsdd} for a similar cubical array.
    An important factor for the performance of DOA estimation methods is the duration of the windows over which they apply the pooling.
    While shorter windows are desired for moving sources, longer ones can improve accuracy as they provide more observations.
    To investigate this trade-off, we evaluated all three methods with moving windows extracted from datasets with fixed $T_{60}$ and
    five different durations, $L_{\text{win}}=0.05$, 0.2, 0.5, \SI{1.0}{s}, and "Full" referring to the case when the full sentence
    (\textasciitilde\SI{3}{s}) was used for the pooling.
    As shown in Table~\ref{table:comparison}, all three methods perform better with longer window lengths as the reverberation in the
    environment increases.
    However with an increase in window lengths, performance of LSDD and SRP-PHAT eventually drops while C-SL consistently performs better
    and better with more observations becoming available in all four conditions.
    This can be explained by the fact that C-SL applies the pooling in the world frame whereas the other two do that in the sensor frame.
    When the motion of the array and that of source are independent (e.g. stationary sources) directions in the world frame vary slower
    thus C-SL benefits better from longer windows.

    \begin{figure}[t]
        \centering
        \subfloat[]{%
            \hspace*{0mm}
            \raisebox{0mm}{\includegraphics[width=7.9cm, valign=t]{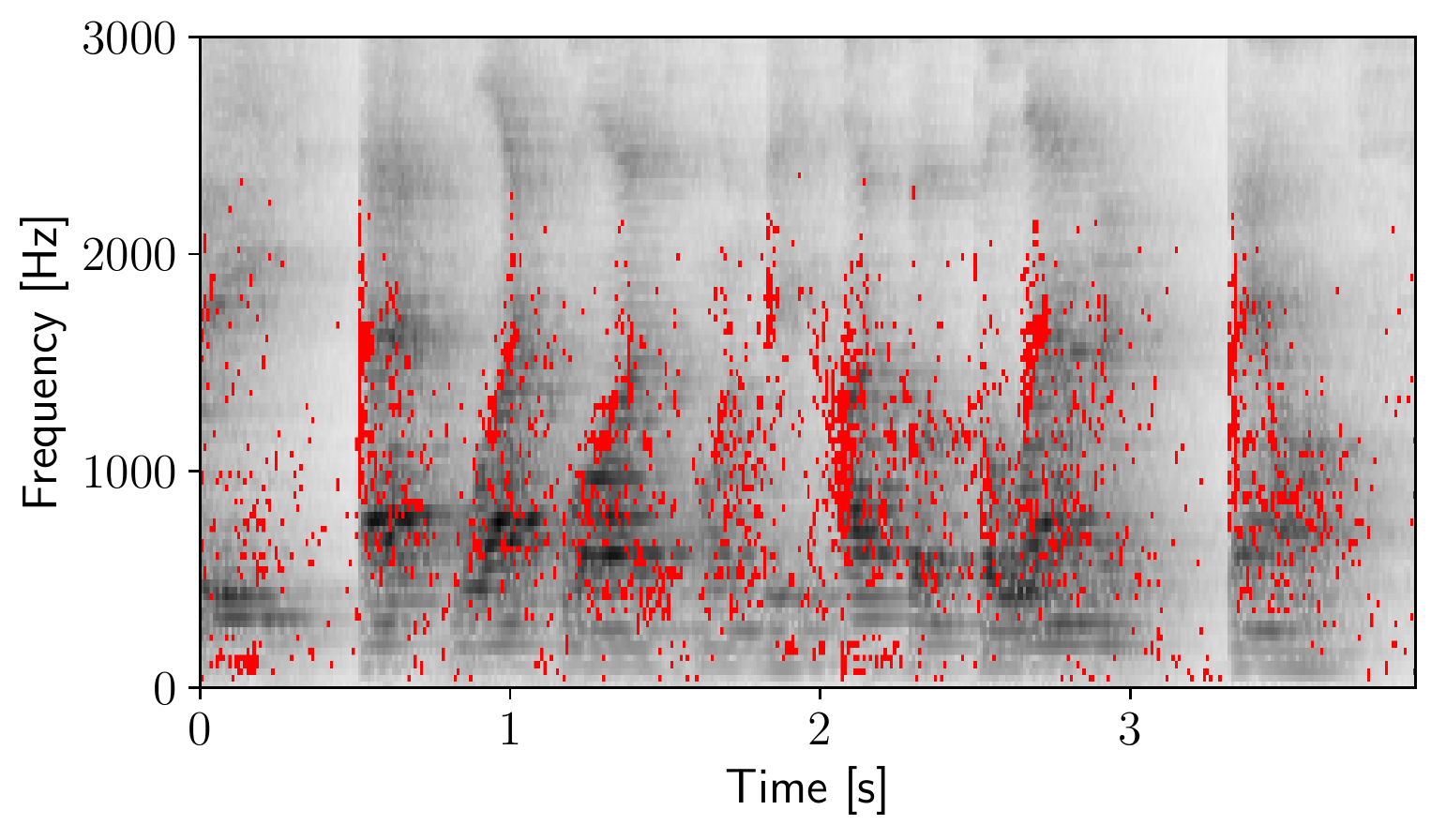}}%
            \label{subfig:scatter}%
        }
        \subfloat[]{%
            \raisebox{0mm}{\includegraphics[width=6cm, valign=t]{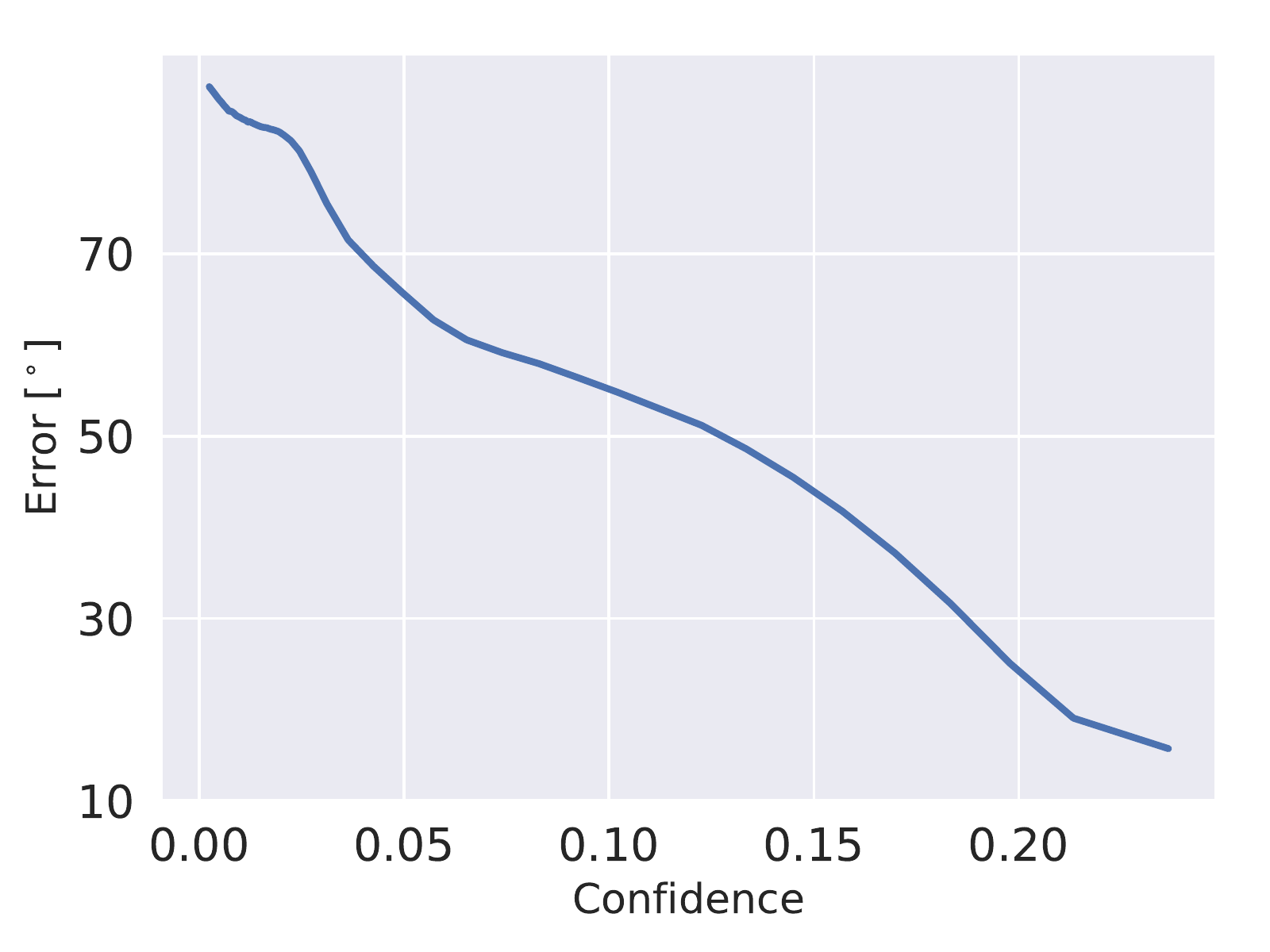}}%
            \label{subfig:err_vs_dom}%
        }
        \caption{Evaluation of confidence weights predicted by C-SL:
            \protect\subref{subfig:scatter} Scatter plot of time-frequency bins (red pixels) with estimated confidence in the
            95-th percentile overlaid on a sample spectrogram. \protect\subref{subfig:err_vs_dom} Sensor-frame direction
            errors $\sigma(r^s, \dblacc\hat{\tilde{r}}^s)$ vs.\ predicted confidence $\norm{\tilde{r}^s}$ estimated by
            C-SL averaged over percentile groups of time-frequency bins in the $C_\text{mixed}$ dataset.}
        \label{fig:dominance}
    \end{figure}%

    In the second experiment, we performed an assessment of the confidence weights predicted by C-SL\@.
    To find these weights, we calculated $\ell_2$ norm of the MLP output for all time-frequency bins extracted from the test dataset in
    $C_\text{mixed}$ condition.
    Figure~\ref{fig:dominance}\subref{subfig:scatter} illustrates the scatter plot of the bins with weights in the 95-percentile for a
    highly reverberated session in this dataset ($T_{60}=\SI{750}{ms}$) overlaid on the session spectrogram.
    As expected, majority of these bins are concentrated around signal onsets at frequencies carrying higher energy.
    We also examined how these confidence scores are related to sensor frame estimation errors calculated by
    $\sigma(r^s, \hat{\tilde{r}}^s)$.
    Figure~\ref{fig:dominance}\subref{subfig:err_vs_dom} shows error curve vs.\ confidence weights found using a quantile-based binning
    of time-frequency bins.
    The monotonic decrease in average errors indicates that the model has successfully learned to predict uncertainty in estimations.

    In our last experiment, we investigated application of C-SL at inference time in a multi-source environment.
    In such conditions, it can be assumed that each time-frequency bin is dominated by one source, therefore finding location of the
    sources can be cast as a clustering task where time-frequency bins are assigned to different clusters based on their predicted world
    frame direction.
    While many approaches could be used for the clustering, we opted for a non-parametric kernel density estimation (KDE) based
    technique to detect dominant directions.
    In this method, given estimated world-frame directions for an ensemble of bins within a window, we first
    approximate the weighted density of directions on a uniform grid (same as the one used in the first experiment) by
    $\psi_n(r_j)=\sum\limits_{k,n'\in I_n}\norm{\tilde{r}^w_{k,n'}}e^{-\sigma(r_j,\hat{\tilde{r}}^w_{k,n'})/\alpha}$
    where we set $\alpha=\ang{1}$ as the bandwidth of the kernel.
    Given the maximum number of sources $N_\text{src}$, we then find local maxima of function $\psi$ on the grid, and pick $N_\text{src}$
    peaks with highest density as the estimates of source direction for the window.

    \begin{wraptable}[25]{r}{7cm}
        \centering
        \begin{tabular}{c}
            \begin{minipage}{0.5\textwidth}
                \centering
                \includegraphics[width=7cm]{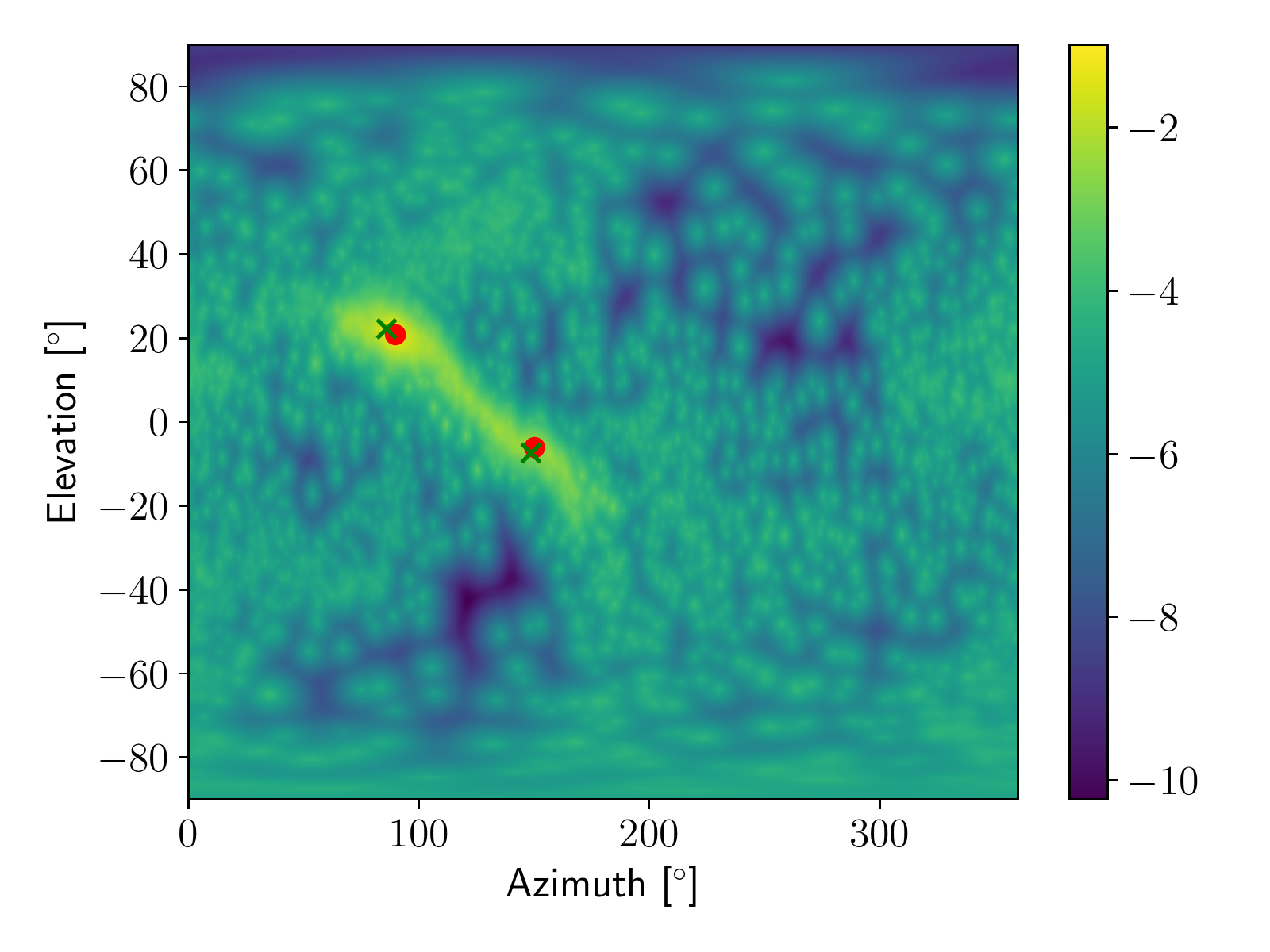}
                \captionsetup{width=.8\linewidth}
                \captionof{figure}{Log-scale kernel density values estimated by C-SL for a sample time window of length \SI{200}{ms}.
                    ($\textcolor{red}{\Bigcdot}$) and ($\textcolor{ForestGreen}{\pmb{\times}}$) depict predicted and ground truth source
                    locations on the grid.}
                \label{fig:density}
            \end{minipage} \\
            \\
            \begin{minipage}{.5\textwidth}
                \captionsetup{width=.8\linewidth}
                \caption{C-SL estimation errors (in degree) in two-speaker scenario computed for different window lengths in
                    anechoic condition.}
                \label{table:2_src}
                \centering
                \begin{tabular}{@{}lcc@{}}
                    \toprule
                    $L_{\text{win}}(\si{s})$ & $d_{\text{w-chamfer}}$ \\
                    \midrule
                    0.05 & $\CI{9.87}{0.11}$ \\
                    0.2 & $\CI{4.26}{0.13}$ \\
                    0.5 & $\CI{2.85}{0.15}$ \\
                    1.0 & $\CI{2.22}{0.15}$ \\
                    \bottomrule
                \end{tabular}
            \end{minipage}
        \end{tabular}
    \end{wraptable}

    For this experiment, we created a new test dataset in anechoic condition consisting of 600 sessions each with two sources.
    To generate each session, two sentences were randomly selected from the test split we generated before
    and spatialized according to random independent locations selected for each of the sources and motions of the array similar to the single
    source condition.
    Spatialized sounds from the two sources were then added together (after padding the shorter one with zeros at the end) to generate
    the recordings at each microphone.
    Finally, time windows extracted from these recordings with at least one source present were used to evaluate the performance of C-SL
    in a two-speaker condition.
    Figure~\ref{fig:density} shows log-scale grid densities calculated for a sample input window and the corresponding predicted and
    ground truth pairs of directions.
    As it can be seen, the identified peaks are very sharp and lie close to ground-truth locations of the sources, a trend that we
    found to be generally true when using \@C-SL\@.
    In order to quantify the error between the predicted set of sensor frame directions and their ground truth values, denoted by
    $\bar{R}^s$ and $R^s$, we used a weighted version of Chamfer distance to match directions in the two sets and measure the deviation
    between them as following:

    \begin{align*}
        d_{\text{w-chamfer}}(\bar{R}^s, R^s)=&\frac{1}{\lvert R^s\rvert}\sum\limits_{r\in R^s}\min\limits_{r'\in \bar{R}^s}\sigma(r,r') \\
        +&\frac{1}{\sum\limits_{r'\in \bar{R}^s}\psi(r')}\sum\limits_{r'\in \bar{R}^s}\psi(r')\min\limits_{r\in R^s}\sigma(r,r')\qquad(8)
        \label{eq:chamfer}
    \end{align*}

    \newpage
    The weighting can be thought of as a mechanism to filter out spurious peaks based on their density without having to choose thresholds.
    We calculated values of this metric for different window durations ranging from \SIrange{50}{1000}{ms}.
    Results, shown in Table~\ref{table:2_src}, demonstrate that in conjunction with the appropriate clustering scheme C-SL can also be
    utilized in multi-source environments.

    \section{Conclusion}\label{sec:conclusion}

    In this paper, we presented Contrastive Sound Localization (C-SL), a framework for learning acoustic-spatial mappings
    from unlabled data collected by microphone arrays of arbitrary geometry.
    C-SL combines contrastive learning and acoustic-inertial sensor fusion to simultaneously calibrate the array and estimate DOAs
    in a self-supervised manner.
    Our evaluations demonstrate that, by leveraging array movements, C-SL can localize sounds in a wide range of conditions with no
    additional information about the array or the sources available
    The relaxed data collection, simplicity and low computational requirements to train the model, together with the encouraging results in
    challenging conditions are advancements offered by C-SL that pave the way toward personalized hearing applications.

    \bibliography{reference.bib}

    \newpage
    \section*{Appendix}\label{sec:app}
    \setcounter{equation}{0}
    \begin{theorem*}
    Given bijective functions $g,f\colon A\rightarrow\mathbb{S}^2$ defined on the non-empty set $A$ and the constraint $C$:
    $\forall x,y\in A, \forall R_x,R_y\in SO(3)\colon R_xf(x)=R_yf(y)\iff R_xg(x)=R_yg(y)$, $C$ holds if and only if $f=\pm g$.
    \end{theorem*}

    \begin{proof}
     It is trivial to show $C$ holds when $f=\pm g$.\\
     Now let $\theta(R)$ denote the rotation angle corresponding to rotation matrix $R$.
     It can be shown that:

     \begin{equation}
         \forall u,v\in \mathbb{S}^2 \colon \langle u,v\rangle=\max\limits_{R\in SO(3)\colon Ru=v} \cos{(\theta(R))}
         \label{eq:dot}
     \end{equation}

     Using this we can show that if $C$ holds:

     \begin{equation}
         \forall x,y\in A: \langle f(x),f(y)\rangle=\langle g(x),g(y)\rangle
         \label{eq:iso}\\
     \end{equation}

     By setting values of $y$ in~\eqref{eq:iso} to $a_1=g^{-1}(\uvec{i})$, $a_2=g^{-1}(\uvec{j})$, and $a_3=g^{-1}(\uvec{k})$
     we get:

     \begin{equation}
         g=Pf
         \label{eq:lin}
     \end{equation}

     where

     \begin{equation}
        P=\left[\begin{array}{@{}c|c|c@{}}
        f(a_1) & f(a_2) & f(a_3)
        \end{array}\right]^T
        \label{eq:P}
     \end{equation}

     Furthermore since~\eqref{eq:lin} also holds for $a_1$, $a_2$ and $a_3$ it can be shown that:

     \begin{equation}
         P^{T}P=PP^T=I
         \label{eq:ortho}
     \end{equation}

     By plugging~\eqref{eq:lin} into $C$ and setting $R=R_y^{-1}R_x$, we will get:

     \begin{equation}
         \forall x,y\in A, \forall R\in SO(3): Rf(x)=f(y)\iff RPf(x)=Pf(y)
     \end{equation}

     which is equivalent to

     \begin{equation}
         \forall R\in SO(3): RP=PR
         \label{eq:transd}
     \end{equation}

     ~\eqref{eq:ortho} and~\eqref{eq:transd} imply that $P=\pm I$.
     Therefore $f=\pm g$.
    \end{proof}

\end{document}